\documentclass[reqno]{amsart}

\usepackage[margin=3cm]{geometry} 

\usepackage{bm}
\usepackage{hyperref} 
\hypersetup{
    colorlinks = true,
    linkcolor={red!50!black},
    citecolor={blue!50!black},
    urlcolor={blue!50!black}
}
\usepackage{cite}

\usepackage{siunitx}
\usepackage{mathtools}
\usepackage{amssymb}
\usepackage{mathrsfs}
\usepackage[normalem]{ulem}
\usepackage{enumitem}
\usepackage{amsthm} 
\usepackage[capitalise]{cleveref}

\usepackage[dvipsnames]{xcolor}
\usepackage{tcolorbox} 

\usepackage{xparse} 

\newcommand{\be}{\begin{equation}}
\newcommand{\ee}{\end{equation}}
\newcommand{\e}[1]{\label{#1}\end{equation}}
\newcommand{\bea}{\begin{eqnarray}}
\newcommand{\eea}{\end{eqnarray}}
\newcommand{\ea}[1]{\label{#1}\end{eqnarray}}
\newcommand{\bse}{\begin{subequations}}
\newcommand{\ese}{\end{subequations}}

\newcommand{\simleq}{\; \raisebox{-0.4ex}{\tiny$\stackrel{{\textstyle<}}{\sim}$}\;}

\newcommand{\myeq}[2]{\stackrel{\mathclap{\mbox{\tiny #1 }}}{#2}}

\newcommand{\iu}{i\mkern1mu}
\newcommand{\func}[2]{\operatorname{#1}\left[#2\right]} 
\newcommand{\funcwrt}[3]{\operatorname{#1}_{\mathrm{#3}}\left[#2\right]} 
\newcommand{\setfunc}[2]{\left\{\operatorname{#1}\left[#2\right]\right\}} 
\newcommand{\funclims}[3]{\operatorname*{#1}\limits_{#2}\left[#3\right]} 

\newcommand{\ket}[1]{\left| #1 \right>}
\newcommand{\bra}[1]{\left< #1 \right|}

\newcommand{\eval}[1]{\mathbb{E}\left[#1\right]} 

\newcommand{\est}[1]{\hat{#1}} 
\newcommand{\Op}[1]{\mathbf{#1}} 

\newcommand{\close}[1]{\overline{#1}} 

\newcommand{\extreal}{\close{\mathbb{R}}} 

\newcommand{\norm}[1]{\lVert #1 \rVert} 

\NewDocumentCommand{\setvalO}{ O{O} O{I} }{\setfunc{\text{\ensuremath{#1}} \mapsto Val}{#2}} 

\NewDocumentCommand{\ID}{ O{I} o } 
{ \IfNoValueTF {#2}
{\ensuremath{(#1,\setfunc{Val}{#1})}}
{\ensuremath{(#1,\setfunc{Val}{#2})}}
}

\usepackage{titlesec} 
\titleformat{\section}{\centering\Large\sffamily}{\thesection.\quad}{0em}{}{}
\titlespacing*{\section}{0pt}{*4}{*1}
\titleformat{\subsection}{\raggedright\large\bf\color{blue!30!black}}{\thesubsection.\quad}{0em}{}{}
\titlespacing*{\subsection}{0pt}{*3}{*1}
\titleformat{\subsubsection}{\raggedright\normalfont\bf}{\thesubsubsection.\quad}{0em}{}{}
\titlespacing*{\subsubsection}{0pt}{*2}{*0.5}

\newtheorem{proposition}{Proposition}
\newtheorem{theorem}{Theorem}

\newcommand{\thistheoremname}{}
\newtheorem{genericthm}[theorem]{\thistheoremname}

\theoremstyle{definition}
\newtheorem{definition}{Definition}

\newtheorem{remark}{Remark}
\newtheorem{assumption}{Assumption}

\newtheorem{question}{Question}

\crefname{assumption}{Assumption}{Assumptions}

\NewDocumentEnvironment{genthm}{m m}{%
    \begin{#1}[#2]%
}{%
    \end{#1}%
}

\makeatletter
\newcommand{\myitem}[1]{%
\item[#1]\protected@edef\@currentlabel{#1}%
}
\makeatother

\begin{document}
\title{The Quantum Cram\'er-Rao lower bound\\ (Why quantum computers won't work I)}
\maketitle

\begin{centering}
\author{Liam P. McGuinness}\\
\address{Institute for Quantum Optics, Ulm University, 89081, Ulm, Germany}\\
\email{Email: \href{mailto:liam@grtoet.com}{liam@grtoet.com}}\\
\end{centering}

\begin{abstract}
Quantum information science currently poses a troubling contradiction. It can be summarized as:
\begin{enumerate}
\item To factor efficiently, quantum computers must perform exponentially precise energy estimation.
\item Exponentially precise energy estimation is impossible according to both the Heisenberg limit and the Cram\'er-Rao lower bound in quantum metrology.
\end{enumerate}
It is surprising that such a dramatic contradiction exists between two accepted predictions of quantum mechanics, and yet this contradiction it is not widely discussed. It is even more surprising when one notes it is not a minor discrepancy – the two statements differ by an exponential margin. Not only that, whether (1) or (2) is correct is of fundamental importance to the realisation of an important class of quantum technologies. If (2) is correct, then quantum computers are much less powerful than expected. This work resolves the above contradiction by defining a computational model in which a wide range of computational problems are not solvable in polynomial time. We then show that this computational model applies to the majority of quantum algorithms, including Shor's algorithm.
\end{abstract}

\section{Never the twain shall meet}
The contradiction noted in the abstract tells a story of two different worlds, or more aptly, two contradictory stories told concurrently by these different worlds. Asked the following question:
\begin{quotation}
With what error does quantum mechanics predict one can estimate the unknown period $L$ of a Hamiltonian; with access to $n$ qubits, and total time $t$?
\end{quotation}
The answers given by members of the quantum metrology/sensing and the quantum computing scientific communities, differ by an amount which is frankly disconcerting\footnote{That they should differ at all is a cause for concern if we expect a mathematically rigorous answer, but that they differ by an exponential margin!}. Before presenting their answers, let's first rephrase the question to a form more familiar to those in quantum sensing:
\begin{quotation}
With access to $n$ qubits and total time $t$, with what uncertainty does quantum mechanics predict one can estimate the unknown frequency $\omega$ of an applied Hamiltonian\footnote{In the literature $\Op{H}(\omega, t)$ is generally assumed to have sinusoidal time-dependence, but we can rephrase more generally to estimate any Fourier component of an arbitrary Hamiltonian.}?
\end{quotation}
Apart from a constant factor (to account for the change of units), the two questions are identical -- therefore we should expect the answers to the questions to be identical. In terms of frequency precision, the answers are
\begin{enumerate}
\myitem{(A1)}\label{item:A1} \textbf{Quantum computing:} With high probability, we can obtain an estimate of the frequency that has an error $\epsilon \simeq O(2^{-\func{poly}{n,t}}) \simeq O(\func{exp}{-\func{poly}{n,t}})$\footnote{Here we are using `big Oh' notation, and $\func{poly}{n,t}$ is a non-negative function polynomial in $n$ and $t$.}.
\myitem{(A2)}\label{item:A2} \textbf{Quantum metrology/sensing:} No estimate of the frequency, $\est{\omega}$ can have an uncertainty lower than $\Delta \est{\omega} \geq \frac{1}{nt}$. This limit is known as the Heisenberg limit (HL) in quantum metrology and derives from a theorem known as the quantum Cram\'er-Rao lower bound.
\end{enumerate}


Whilst the mathematical terminology used by each community to express their answer differs slightly, it should be clear that the two answers are fundamentally incompatible. With a little work to redefine the probabilistic error of a quantum computation in terms of the \hyperlink{Mean squared error}{mean squared error}, and defining the \hyperlink{uncertainty}{uncertainty} as the square-root of the mean squared error, one can show that the answers differ by an exponential margin\footnote{This requires that the range of possible frequencies is bounded, i.e. $a \leq \omega \leq b$ for finite $a,b$. Then as the probability to obtain an estimate with low error can be arbitrarily increased whilst maintaining exponential scaling, the mean squared error must reduce exponentially, even allowing a low probability to obtain an estimate with large error.}.

That such a striking contradiction exists can be readily checked by anyone who reads the literature, speaks to members of these communities or queries AI using the above pointed line of questioning. Supposedly, quantum computers can solve a range of problems, not only finding the unknown frequency/period, but also the unknown energy eigenvalues, angle of rotation or phase of a Hamiltonian (equivalently Unitary) with a precision that improves exponentially in time. The exponential improvement over (known) classical algorithms requires entanglement between $n$ qubits in the computational register and is sometimes expressed as a precision improving exponentially with $n$. In each case there is a conflict with established bounds in quantum metrology. In quantum sensing, the best precision one can estimate those same parameters in a Hamiltonian, even with entanglement of $n$ qubits, improves linearly with time (and with $n$).

Put bluntly, these two different fields of quantum science are in gross disagreement. From a sociological perspective it is interesting to ask how such a disagreement can persist for so long without being discussed, but our focus here will be on the mathematical issues.

Before presenting work to resolve this contradiction, it is worth reassuring many readers that the above statements are \emph{in essence} correct; that I am not misrepresenting the work of either community, or leaving out important qualifying details that would significantly change the answer. To ease that task, the next section gives a more complete literature survey showing that the claims presented in \ref{item:A1}, \ref{item:A2} above are indeed a fair summary of the current state of understanding in the two fields.

\section{Efficient quantum computation implies exponentially precise estimation}\label{sec:QC}

Whilst it is uncommon in computer science to analyse the performance of an algorithm keeping the input length/problem length fixed whilst varying the number of qubits, if we make this adjustment we can readily compare computational performance to the Heisenberg limit. In particular, assume the unitary evolution of $n$ qubits depends on some parameter $\theta$ taking on a fixed value in the interval $0 \leq \theta < 1$. We would like to estimate the value of $\theta$ in an allocated time $t$ and with access to $n$ qubits, and we are interested in how the estimation precision improves with $t$ and $n$.

The class of efficient computational solutions to this problem have a precision that improves roughly exponentially in time. This statement, is in effect, simply the definition of an efficient algorithm. To see this in detail, note that an efficient quantum computation returns the value of $\theta$ to $n$ bit precision in a time that increases polynomially with $n$, giving a performance
\[ \Delta\est{\theta} \sim 1/2^n \quad\quad \myeq{$t=\func{poly}{n}$}{\sim}\quad\quad 1/2^{\func{poly}{t}} \sim 2^{-t}.\] 
Expressing $\theta$ in decimal units, increasing $n$ by one, implies a precision improvement by a factor of 10, for a marginal increase in time.

A special case of this problem is quantum phase estimation, where unitary evolution is parametrized by $\Op{U}(2\pi \iu \varphi)$. Any good textbook on quantum computation will confirm that the unknown parameter $\varphi$, called a quantum phase, can be estimated with a precision that improves exponentially in time. Some examples are
\begin{enumerate}
\item Kitaev, Shen and Vyalyi's ``Classical and Quantum Computation" \cite{Kitaev2002}, in particular \S\,13.5.3 ``Determining the phase with exponential precision". If the title of the section didn't already give the game away, you can read on to where KSV note for the quantum phase estimation problem, their algorithm ``allows us to determine $\varphi$ with precision $1/2^{2n+2}$ \emph{efficiently} in linear time with constant memory".
\item Nielsen and Chuang's ``Quantum Computation and Quantum Information" \cite{Nielsen2000}, in particular Ch.\,5.2 ``Phase estimation" \S\,5.2.1 ``Performance and requirements". Nielsen and Chuang provide error analysis of the quantum fourier transform when applied to quantum phase estimation, showing that $\varphi$ can be estimated to a precision $\Delta \est{\varphi} = 2^{-n}$ in time $t = O(n^3)$, with access to slightly more than $2n$ qubits.
\item John Preskill's ``Lecture Notes for Physics 229: Quantum Information and Computation" \cite{Preskill2015} (which can be found online). In Ch.\,6.2 ``Periodicity" Preskill notes that the quantum fourier transform can find the period $L$ of a function with a precision $\Delta \est{L} = 2^{-n}$ in $\func{poly}{n}$ time. ``Our quantum algorithm can be applied to finding, in poly(n) time, the period of any function that we can compute in poly(n) time. Efficient period finding allows us to efficiently solve a variety of (apparently) hard problems, such as factoring an integer, or evaluating a discrete logarithm." In \S 6.4 ``Phase estimation" Preskill then connects period finding to quantum phase estimation, so exponentially fast period finding implies the quantum phase $\varphi$ can be measured with exponential accuracy.
\item For those who prefer videos to books, I recommend Ryan O'Donnell's ``Quantum Computer Programming in 100 Easy Lessons" on \href{https://www.youtube.com/playlist?list=PLm3J0oaFux3bF48kurxGR6jrmPaQf6lkN}{Youtube}. In Lessons 60 -- 66 O'Donnell discusses how rotation estimation is related to factoring, and he describes a quantum algorithm for estimation $\theta$ to $n$ digits of accuracy, i.e. $\Delta \est{\theta} \simeq 10^{-n}$ in time polynomial in $n$ (here $\theta$ is the rotation angle).
\end{enumerate}

Beyond the textbooks, a good starting reference to check the veracity of \ref{item:A1} is a 2017 paper by Yosi Atia and Dorit Aharonov \cite{Atia2017}. Atia and Aharonov clearly state that Shor's algorithm implies exponentially precise energy estimation
, and they are not alone, nor the first, in having made this observation. In several works, Berry, Childs, Kothari and co-authors \cite{Berry2014, Berry2015} presented algorithms for Hamiltonian simulation with logarithmic run-time (exponential precision), building in turn on work by Seth Lloyd \cite{Lloyd1996} and others in generating exponential speed-ups \cite{Abrams1999, Lloyd2014}.


One might ask whether I am being selective in the quantum algorithms discussed above, and how this relates to the performance of all quantum algorithms? It turns out that several results indicate exponentially precise measurements are required for all efficient quantum algorithms. In \cite{Berry2015}, BCK note that efficient algorithms for Hamiltonian simulation cover the entire class of efficient quantum algorithms including integer factorization. Their error analysis exponentially violates the uncertainty limit in quantum metrology, and by connection, due to the computational class they establish this implies all efficient quantum algorithms.

From a different perspective, the work “Grand Unification of Quantum Algorithms” \cite{Martyn2021} establishes an equivalence between the performance of quantum algorithms and parameter estimation through a framework called quantum signal processing \cite{Low2017}. Formalising the often made observation in quantum information science -- that quantum algorithms including amplitude amplification, Grover’s search, and Hamiltonian simulation operate in essentially the same manner -- the authors show that the performance of each algorithm is determined by the precision they achieve for estimating the unknown phase angles in pulses (an unknown signal) applied to the computational register.

Finally, that quantum computation and precision measurement are intimately connected is extensively discussed by Childs, Preskill and Renes in \cite{Childs2000}. It is worth pointing out that their discussion is self-contradictory. For example CPR state that the accuracy of the quantum fourier transform for frequency estimation is limited to $\Delta \est{\omega} \geq 1/t$, rather than $\Delta \est{\omega} \sim 1/2^t$ as claimed in several of their other works and indeed later in the same paper; ``The accuracy is limited by an energy-time uncertainty relation of the form $T \Delta \omega \sim 1$" \cite{Childs2000}.




\begin{remark}[Discrete or real-valued]
It is important to note that the performance of these algorithms does not depend on the solution $\theta$ taking on only discrete values, i.e. that $\theta$ can be represented exactly with an $n$-bit binary expansion. The same precision is obtained if $\theta$ is a real number. That is clear in both the Nielsen, Chuang \cite{Nielsen2000} and Kitaev, Shen, Vyalyi \cite{Kitaev2002} analysis. It is important to make this observation, because the contradiction we raise is based on the quantum Cram\'er-Rao lower bound, which as we will later show, explicitly assumes that $\theta$ is a real number. One might therefore argue (incorrectly) that the bound does not apply to quantum computers. Noting that the performance of quantum computers (an exponentially improving precision) remains regardless of whether $\theta$ can be expressed exactly with $n$ bits or not, we can dismiss this objection.
\end{remark}
\subsubsection*{Summary}
The message presented by the quantum computing literature is clear, to solve a range of problems efficiently, quantum computers must estimate a parameter in the Hamiltonian with a precision that improves exponentially in time and/or number of qubits. Expressing these total resources as any polynomial function $\Gamma \sim \func{poly}{n,t}$ (and keeping energy fixed), quantum computers display a precision
\[ \Delta \est{\theta} \sim 2^{-\Gamma},\]
where $\theta$ is an arbitrary parameter in a Hamiltonian applied to the computer. That an efficient computation must exhibit exponential improvement is uncontroversial, since this is the very definition of an efficient algorithm as a function of input length for a problem on a bounded interval. Slightly less well discussed, but also uncontroversial is that quantum algorithms can be recast in terms of parameter estimation. It is this connection in particular that we use to prove a contradiction regarding the performance of quantum computers. A further key component to our argument is the observation that quantum computers achieve exponential precision from a single measurement on an entangled system of qubits, i.e. after a single computational run.




\section{The Heisenberg limit in metrology}\label{sec:HL}
In contrast to the expected performance of quantum algorithms, a theoretical bound exists in the field of quantum metrology which places a much stronger restriction on the uncertainty one can estimate an unknown parameter $\theta$ in a Hamiltonian $\Op{H}(\theta)$ in a given time and with $n$ qubits. Quantitatively, with a single particle ($n = 1$) and measurement time $t$, if one can achieve an uncertainty in estimating $\theta$ of
\[ \Delta \est{\theta} = \lambda/t,\]
then quantum mechanics predicts an uncertainty of 
\be\label{eq:HL} \Delta \est{\theta} \geq \lambda/(nt) \ee
using $n$ entangled particles in the same time. I.e. an uncertainty improvement by a factor of $n$. \Cref{eq:HL} is known as the Heisenberg limit (HL) in quantum metrology, and $\lambda$ is a change of units factor, that converts from inverse time to the units of $\theta$ (actually $\est{\theta}$).

As far as I am aware, there is no debate on the Heisenberg limit in the community, I do not know any physicists that challenge \cref{eq:HL}. You can see this for yourself by looking at any of the references included in \cite{McGuinness2021} or \cite{McGuinness2023b}, by looking at review papers on quantum metrology \cite{Giovannetti2004, Giovannetti2006, Pezze2018, Degen2017} or even the Wikipedia page on \href{https://en.wikipedia.org/wiki/Quantum_metrology}{quantum metrology} (see the section on scaling). Everyone is clear that the precision improves by (at most) a factor of $n$ using entanglement, and linearly with time. There is no argument on this point anywhere in quantum metrology. Note that even some quantum computing experts seem to believe the Heisenberg limit, John Preskill for example, is an author on this paper \cite{Zhou2018}, which explicitly states that the best precision is improves as $1/n$.

Beyond the theoretical analysis, there are numerous experimental tests of \cref{eq:HL} and none of these experimental works claim to achieve a precision beyond that dictated by the HL. Any review article on quantum metrology will reproduce these claims, without actually checking if they are correct\footnote{This is standard procedure for review articles.} \cite{Pezze2018, Degen2017, Schnabel2017}. Or you might be interested in reviews that analyse the claims critically and find that none of the works even get close to the HL \cite{McGuinness2021,McGuinness2023b, ThomasPeter2011, McGuinness2022, McGuinness2023, McGuinness2023a}.

Consider for a moment the counter-factual to the Heisenberg limit, i.e. that efficient quantum computers are realisable. If true, it would mean we could implement measurement techniques with exponentially better precision than currently in use! What then is the entire field of atomic clocks or gravitational wave detection doing; can we really expect that they have left this free lunch to be eaten by someone else? Is it that they aware of such measurement schemes and they simply choose not to implement them, or are they oblivious to the possibility that their painstaking experiments which take years to construct, and which have been analysed for decades can be improved so radically?

\subsubsection*{Summary}
\Cref{sec:QC,sec:HL} summarized a dilemma posed by the quantum science literature. One part of the dilemma is 
\begin{enumerate}
\item To solve a range of problems efficiently (sometimes exponentially faster than the best known classical algorithms), quantum computers estimate the unknown phase, energy, angle, or frequency of a Hamiltonian with a precision improving exponentially with time and number of qubits.
\end{enumerate}
The range of problems include factorisation (Shor's algorithm), Hamiltonian simulation, phase estimation, principle component estimation, Hamiltonian graph problem, hidden subgroup, discrete log... The second part of the dilemma is
\begin{enumerate}
\setcounter{enumi}{1}
\item The Heisenberg limit in quantum metrology restricts the precision attainable from any measurement on $n$ qubits in time $t$. It states that the uncertainty one can estimate an unknown parameter $\theta$ in a Hamiltonian is bounded by $\Delta \est{\theta} \geq \lambda/(nt)$, for constant $\lambda$.
\end{enumerate}

It is clear that these two statements are incompatible, implying that at least one of them is incorrect. In the next section, we go through derivation of the HL in detail. It is based upon a theorem known as the quantum Cram\'er-Rao lower bound.

\begin{remark}[The Heisenberg limit applies to all computation]
The power of the Heisenberg limit is the complete generality of the bound. If for example we can show the Heisenberg limit forbids exponentially accurate factoring, then this results applies to all algorithms. We would have proven that no algorithm can factor efficiently, not just the best known algorithm, no algorithm. While this observation should make one stand up and notice just how powerful the approach can be when applied to computer complexity theory, it should also make one sceptical as to the likelihood that it will work, since it seems to indicate a route to proving $P \neq NP$.
\end{remark}

\section{The (quantum) Cram\'er-Rao lower bound}
\subsection{Classical statistical parameter estimation}
Although often presented in a complicated and dense manner, the Cram\'er-Rao lower bound (CRLB) in statistics is actually straightforward to understand. It tells us the amount of information that a random sample from a probability distribution can contain on some unknown parameter $\theta$. In particular, the information on $\theta$ is bounded by how much, on average, the probability distribution depends on $\theta$. Later we will develop the quantum Cram\'er-Rao lower bound, it is connected to the classical CRLB by the following observation; the outcome of a quantum measurement on a quantum state vector is equivalent to a random sample from a probability distribution\footnote{The Born rule postulate of quantum mechanics.}. But for now we summarize the CRLB for parameter estimation in classical statistics.

\begin{definition}[\hypertarget{classical estimation problem}{Classical estimation problem}]
Consider an arbitrary measuring device which is used to measure a signal and which outputs a data-set $\bm{x} := \{x_1, x_2, \dotsc, x_R\}$, described by a collection of real numbers\footnote{Here the outcome of a single measurement is denoted $x_i$, so $\bm{x}$ is the vector outcome of $R$ measurements.}
. Parametrizing the signal by $\theta$ allows us to define an \emph{estimation problem} as the task of estimating the unknown value of $\theta$, using \emph{only} the information provided by the data-set.
\end{definition}
The CRLB is motivated by the following simple question.
\begin{question}\label{qu:crlb}
How well can one perform the \hyperlink{classical estimation problem}{Classical estimation problem}. I.e. best estimate the unknown value of $\theta$ from the $R$-point data-set outputted by the measurement device.
\end{question}
To give a rigorous answer to this question, some mathematical assumptions and definitions are required.
\begin{assumption}[Scalar, classical parameter]\label{ass:classical}
Whilst not necessary, for this analysis we assume that $\theta$ is a scalar, one-dimensional parameter taking on values in the real numbers. I.e. $\theta \in \Theta$, where $\Theta \subseteq \mathbb{R}$ is called the \emph{parameter space}. We further assume $\theta$ is a \emph{classical} parameter, meaning that it has a definite, fixed value. Whilst we do not know the value of $\theta$, it's value remains fixed from measurement to measurement, and can in principle (given enough resources) be estimated to arbitrary precision. One could say that there is no intrinsic uncertainty associated with the value of $\theta$.

A less discussed caveat to this assumption, is that $\Theta$ be a (Lebesgue) measurable subset of the reals. To avoid the mathematical technicalities involved in defining Lebesgue measurable sets and working with them, we are going to simply assume that $\Theta$ is a single interval in the reals\footnote{Actually, Cram\'er made this same assumption in his original derivation of the bound \cite{Cramer1946}. A single interval is Lebesgue measurable.}. I.e. that $a \leq \theta \leq b$, for $a,b$ real numbers with $a \leq b$.
\end{assumption}

\begin{remark}[Lebesgue measurable]
In the following, all of the sets and functions are assumed to be Lebesgue measurable\footnote{Don't worry too much about this requirement if you are not familiar with measure theory. In the end we assign a \emph{measure} to each estimate, which quantifies how good the estimate is -- namely how far away from $\theta$ the estimate is. Requiring all of the mathematical sets and functions be measurable ensures that we can assign measures to them.}.
\end{remark}

\begin{assumption}[Probabilistic data-set]\label{ass:pdf}
We assume that the response of the measurement device to $\theta$, is perfectly characterised. I.e. for each value of $\theta$, we know the corresponding output of the measurement device, described by a mapping $\func{p}{\cdot} : \Theta \to \mathcal{X}$. We additionally assume that the mapping is given by a probabilistic function, i.e. we allow that the data has some inherent randomness and is described by a probability density function (PDF) $\func{p}{\bm{x};\theta}$, which we call the \emph{measurement PDF}. This is a mapping to a probability space $\func{p}{\cdot} : \Theta \to (\mathcal{X}, \bm{P})$, where $(\bm{x}\in \mathcal{X}, p\in \bm{P})$ denotes the probability to observe the measurement result $\bm{x}$.
We additionally assume the measurement PDF is not many-to-one, such that for two distinct values of $\theta$, $\func{p}{\bm{x};\theta}$ does not map to the same point in probability space. Denoting the outputs $\func{p}{\bm{x};\theta_1} \mapsto (\bm{x_1}, p_1)$, $\func{p}{\bm{x};\theta_2} \mapsto (\bm{x_2}, p_2)$ we have $\theta_1 \neq \theta_2 \implies (\bm{x_1}, p_1) \neq (\bm{x_2}, p_2)$.
\end{assumption}
\begin{remark}
The assumption of a probabilistic data-set makes the estimation problem non-trivial. Since if the measuring device is perfectly characterised and its output is deterministic then we can estimate $\theta$ with no uncertainty. I.e. if the mapping $\func{p}{\cdot} : \Theta \to \mathcal{X}$ is bijective on codomain $\mathbb{R}^R$, then the best estimator is the inverse $\func{p^{-1}}{\cdot}$.
\end{remark}
\begin{definition}[A (posterior) estimate, estimator and error]\label{def:est}
We define a (posterior) \emph{estimator} as a function acting on the data-set which returns a (posterior) \emph{estimate} $\est{\theta} \in \mathbb{R}$ of the value of $\theta$. For a data-set taking on elements in $\mathcal{X}$, the function $\func{est}{\cdot} : \mathcal{X} \to \mathbb{R}$ is the estimator, and its output is the estimate\footnote{Or slightly abusing notation, $\est{\theta} = \func{est}{\bm{x}}$.}
\[ \est{\theta} := a : \func{est}{\bm{x}} \mapsto a.\]
The \emph{error} of the estimate is defined as the distance of $\est{\theta}$ from $\theta$
\[ \func{err}{\est{\theta}} := \func{d}{\est{\theta},\theta}.\]
We will not be strict in distinguishing between  the estimator (the function) and the estimate (the output), using the two interchangeably.
\end{definition}
\begin{remark}
If the data is described by a probability density function, then \cref{qu:crlb} reduces to estimation of probability distributions. Given that $\func{p}{\bm{x};\theta}$ belongs to a parametrized family of probability distributions $\mathcal{F}_{\theta} = \{\func{p}{\bm{x};\theta} : \theta \in \Theta\}$ (parametrized by $\theta$). We can ask, how well can one determine the underlying probability distribution from a given number of samples drawn from the PDF? It turns out that this is (nearly) equivalent to \cref{qu:crlb}.
\end{remark}
At this point I think it is helpful to comment on a conceptual difficulty I have with probability theory and statistical estimation that is not discussed as much as I would like. Until now we have been talking in the past-tense, as if the measuring device has outputted an actual data-set (a collection of real numbers). Given the data-set $\bm{x}$, our task to return a unique estimate, a real number based on that fixed data-set. That is not how the CRLB is actually derived, and the scenario considered is slightly more complicated.

Before the measuring device has even outputted a data-set, we consider the entirety of results the device could \emph{potentially} output. The results lie in the space $\mathcal{X}$, where the probability to observe any result, a point in (or more generally a measurable subset of) $\mathcal{X}$ is given by the PDF $\func{p}{\bm{x};\theta}$. The CRLB works with this probability space, meaning that an estimate is itself a random variable, and the estimator is a function acting on a set, it maps a set of points in a probability space to a new set of points in a probability space. Therefore, we revise \cref{def:est} to define a probabilistic estimate.
\begin{definition}[Probabilistic estimate and estimator]
We define an \emph{estimator} as a function acting on the measurement PDF (the probability space of the data-set) which returns a probabilistic \emph{estimate} $\est{\theta}$ of the value of $\theta$, itself a probability space which we call the \emph{estimate PDF}. The estimator function $\func{est}{\cdot} : (\mathcal{X}, \bm{P}) \to (\est{\Theta}, \bm{P})$ takes the measurement PDF as input, and its output is the estimate PDF describing a random variable taking on real values
\[ \func{est}{\func{p}{\bm{x};\theta}} \mapsto \est{\theta} : \est{\theta} \text{ is a probability space in the reals}.\]
\end{definition}
We now need a new definition of the estimate error. Noting that we can take the expected value and the variance of an estimate\footnote{We do not however treat the parameter $\theta$ as a random variable, see \cref{ass:classical}.}, we use the following 

\begin{definition}[\hypertarget{Mean squared error}{Mean squared error}]
The (Euclidean) \emph{mean squared error} (MSE) of an estimate, with respect to $\theta$, is defined as
\be \func{MSE}{\est{\theta}} := \eval{\left(\est{\theta} - \theta \right)^2},\ee
which is the average squared difference of the estimate from the true value of $\theta$. Here, the expectation is taken with respect to $\func{p}{\bm{x};\theta}$, so in general the mean squared error depends on the value of $\theta$.
\end{definition}
Equivalently, the MSE can be expressed as the sum of the \hyperlink{estimate variance}{estimate variance} and squared bias
\begin{align*} \func{MSE}{\est{\theta}} &= \eval{\left(\est{\theta} - \theta \right)^2} = \eval{\left((\hat{\theta}\textcolor{gray}{-\eval{\hat{\theta}})+(\eval{\hat{\theta}}}-\theta)\right)^2}\\
& = \func{Var}{\hat{\theta}} + \left(\eval{\hat{\theta}} - \theta\right)^2 = \func{Var}{\hat{\theta}} + \left(\func{Bias}{\hat{\theta}}\right)^2,
\end{align*}
where again, both the estimator variance and bias, in general depend on the value of $\theta$.


While the first two assumptions allow us to build up a mathematical framework, which we can ultimately use to bound the error of any estimate of $\theta$, they are not necessary unless we want the bound to be saturated. For example, if the measurement PDF is many-to-one, meaning two different values of $\theta$ give the same measurement result, then the effect of this is to increase the estimation uncertainty of $\theta$. Likewise if we allow that $\theta$ itself is a random variable, then again we increase the estimation uncertainty, since we have added some intrinsic uncertainty to $\theta$ which is combined with our measurement uncertainty.

The next two assumptions do not have these properties. In particular, relaxing the following assumptions can allow for estimators that violate the CRLB.
\begin{assumption}[No additional information]\label{ass:prior}
The (posterior) estimate is obtained only using information provided by the data-set $\bm{x}$. Mathematically, if the data-set takes on elements in $\mathcal{X}$, then we assume the posterior $\func{est}{\bm{x}}$ is a function with a domain restricted to $\mathcal{X}$, and with no dependence on elements of any set outside $\mathcal{X}$. More fully, for a probabilistic estimate, we assume the estimator is a function only of the measurement PDF, i.e. has a domain restricted to $(\mathcal{X}, \bm{P})$.
\end{assumption}
While obvious; clearly if we obtain extra information on $\theta$ from an external source we can produce a better estimate of its value, it turns out that enforcing \cref{ass:prior} is critical and is often violated by quantum algorithms. Ensuring that no extra information is snuck into the analysis is in fact a difficult task.

\begin{assumption}[Regularity conditions]\label{ass:regularity}

We assume the measurement PDF satisfies the following conditions \cite{Kay1993, Nielsen2013}
\begin{itemize}
\item The support $\{ \bm{x} : \func{p}{\bm{x};\theta} > 0 \}$ is identical for all $\func{p}{\bm{x};\theta} \in \mathcal{F}_{\theta}$ (and thus does not depend on $\theta$). I.e. the domain of the measurement PDF for which the probability is non-zero does not depend on $\theta$.
\item The above condition generally ensures that we can swap the order of integration and differentiation if $\bm{x}$ is a continuous random variables. I.e. that $\int_{\mathcal{X}} \func{p}{\bm{x};\theta} \mathrm{d}\bm{x}$ can be differentiated under the integral sign with respect to $\theta$\footnote{For a discrete random variable, we can change the order of summation and differentiation.}. But just to be sure, and because it helps to realise that operation is allowed, I have included it as an extra assumption. Under this assumption one can show
\be \eval{\frac{\partial \func{log}{\func{p}{\bm{x};\theta}}}{\partial \theta}} = 0, \quad \quad \forall \theta \in \Theta, \ee
where the expectation is taken with respect to $\func{p}{\bm{x};\theta}$.
\item The gradient $\frac{\partial \func{p}{\bm{x};\theta}}{\partial \theta}$ exists. This ensures that the output of the measuring device is a measurable set (a probability space).
\end{itemize}
\end{assumption}
Under these assumptions, we can state the Cram\'er-Rao lower bound in terms of the estimate mean squared error.

\begin{genthm}{theorem}{Cram\'er-Rao lower bound}\label{th:CRLB}
Suppose \crefrange{ass:classical}{ass:regularity} hold, then the mean squared error of any estimate of $\theta$ must satisfy 
\be\label{eq:CRLB} \func{MSE}{\est{\theta}} \geq \frac{\left(\frac{\partial }{\partial \theta}\eval{\est{\theta}}\right)^2}{\eval{\left( \frac{\partial\func{log}{\func{p}{\bm{x};\theta}}}{\partial \theta}\right)^2}} \; \myeq{Ass.\,\ref{ass:regularity}}{=}\; \frac{\left(\frac{\partial }{\partial \theta}\eval{\est{\theta}}\right)^2}{-\eval{\frac{\partial^2 \func{log}{\func{p}{\bm{x};\theta}}}{\partial \theta^2}}}, \quad \forall \theta \in \Theta. \ee
where the expectation is taken with respect to $\func{p}{\bm{x};\theta}$.
\end{genthm}
\begin{proof}
See Appendix\,\ref{proof:CRLB}.
\end{proof}

\subsection{The Cram\'er-Rao lower bound and Measure theory}
Answering \cref{qu:crlb} requires that we come up with a way to ``measure" how good an estimate is and the Cram\'er-Rao lower bound takes the mean squared error of an estimate as a starting point for this measure. However, it is not yet a completely satisfying answer. As it currently stands, one way of finding a good estimate is to set, independent of the data-set, $\est{\theta} = \alpha$ for any $a \leq \alpha \leq b$, i.e. simply choose a number in $\Theta$ as the estimate. Since $\frac{\partial }{\partial \theta}\eval{\est{\theta}} = 0$, the CLRB for this estimate is zero, and we can see that $\func{MSE}{\est{\theta}} = 0$ when $\theta = \alpha$, and the bound is obtained.

Computationally, \cref{th:CRLB} bounds the performance of any algorithm in terms of its best case performance on any input. An algorithm which outputs a constant string is considered exceptionally good since it returns the solution for one input string\footnote{Assuming this is true, i.e. that there exists a solution given by the constant string.}. Again, this seems like an unsatisfactory measure for computational performance, and is a loop-hole that we would like our analysis to avoid.

Before introducing one way to resolve this issue, we are going to use the mathematical framework of measure theory to discuss the CRLB in more detail. Roughly, a measure on a set of points in a topological or metric space is the assignment of a non-negative number to the set that characterizes its ``size".

\begin{definition}[Measure -- informal\footnote{See the Appendix for a formal definition of a \hyperlink{Measure}{measure}.}]
A \emph{measure} on a set has the following properties:
\begin{enumerate}
\myitem{($\mu$0)}\label{item:mu0} The measure of any set is non-negative. I.e. there are no a negative sizes.
\myitem{($\mu$1)}\label{item:mu1} The empty set has zero measure, i.e. has zero size.
\myitem{($\mu$2)}\label{item:mu2} The measure of the union of two disjoint sets is equal to the sum of their individual measures. I.e. The size of two separate sets is simply the sum of the size of each set.
\myitem{($\mu$3)}\label{item:mu3} The measure of the union of any enumerable collection of disjoint sets is the sum of their individual measures. I.e. we can extend \ref{item:mu2} to a countable union of disjoint sets.
\end{enumerate}
\end{definition}

We would like to use these axioms \ref{item:mu0}\,--\,\ref{item:mu3} in assigning a measure to an estimate, where the measure reflects the amount of \emph{information} on $\theta$ that the estimate provides. I.e. how good an estimate is (its ``size") is depends on how much information it provides on $\theta$, rather than the mean squared error\footnote{The reason being that the MSE is not a measure. We will see that it is however the reciprocal of a measure -- the Fisher information measure.}. Intuitively, we should expect that information has the properties of a measure -- it is non-negative; no data corresponds to zero information; and if we are given two unrelated pieces of information we expect the total information received to be the sum of the two pieces of information.

Measure theory makes this assignment mathematically rigorous, avoiding non-obvious technical pitfalls\footnote{For instance some sets are not measurable.}, and giving a unique and unambiguous measure to any suitably well-defined set (up to a scaling or normalisation factor). I.e. once we assign a single non-zero measure to any non-empty set, then this defines the measure on all sets in the measure space (see for example Tao \cite{Tao2011} Exercise 1.2.23)\footnote{However, the measure is only unique for a given distance metric. If we can assign a different metric between points, then we have greater freedom in defining the measure. Tao assumes Euclidean distance in his definitions of measure for Ch.\,1.}. In short, measure theory says that, up to a scaling, there is only one way to quantify the amount of information provided by an estimate.

We will see that the denominator in \cref{eq:CRLB} is a measure (of the information provided on $\theta$), and because of its importance to many areas of mathematics and analysis, it has its own name.

\begin{definition}[Fisher information]
The \emph{Fisher information} $\func{I}{\func{p}{\bm{x};\theta}}$ on $\theta$ provided by a PDF $\func{p}{\bm{x};\theta}$ is defined as
\be \func{I}{\func{p}{\bm{x};\theta}} := \eval{\left(\frac{\partial \func{log}{\func{p}{\bm{x};\theta}}}{\partial \theta}\right)^2},\ee
with two qualifiers for edge cases. Namely
\[ \func{I}{\func{p}{\bm{x};\theta}} := 0,\]
when $\func{p}{\bm{x};\theta} = \varnothing$ (i.e. the null data-set) and when $\func{p}{\bm{x};\theta} = 0$ for all $\bm{x}$ in an interval (more generally metric segment) of $\mathcal{X}$\footnote{Alternatively we can remove points in $\mathcal{X}$ with zero probability from the analysis, avoiding the issue of defining $\left(\frac{\partial \func{log}{\func{p}{\bm{x};\theta}}}{\partial \theta}\right)^2 \func{p}{\bm{x};\theta}$ or $\left(\frac{\partial \func{p}{\bm{x};\theta}}{\partial \theta}\right)^2 \frac{1}{\func{p}{\bm{x};\theta}}$ when $\func{p}{\bm{x};\theta} = 0$. This removal does not effect the Fisher information or Cram\'er-Rao lower bound since these points have zero information.}.
\end{definition}
In the context of the \hyperlink{classical estimation problem}{Classical estimation problem}, we can equally refer to the Fisher information on $\theta$ provided by a series of measurements (that produce a data-set $\bm{x}$ with PDF $\func{p}{\bm{x};\theta}$), i.e. the Fisher information on $\theta$ provided by the measuring device.

\begin{proposition}[Fisher information is a measure with respect to $\theta$ on probability distributions]\label{prop:FI}
Let $\bm{x}$ denote the outcome of a countable sequence of measurements with a PDF given by $\func{p}{\bm{x};\theta}$\footnote{I.e. the random variable $\bm{x}$ has a PDF $\func{p}{\bm{x};\theta}$, such that $\func{p}{\bm{x};\theta}$ is an element (each forming a probability space) in family of probability distributions parametrized by $\theta$.}. Then the Fisher information on $\theta$, $\func{I}{\func{p}{\bm{x};\theta}}$, satisfies the \hyperlink{Measure}{measure axioms}.
\end{proposition}
\begin{proof} See Appendix\,\ref{app:measure} for a \hyperlink{Proof}{proof}.
\end{proof}

\begin{definition}[Information and \hypertarget{uncertainty}{uncertainty}]
We will often find it useful to refer to the square root of the Fisher information. Therefore, we define the \emph{information} on $\theta$, $\func{\mathcal{I}}{\func{p}{\bm{x};\theta}}$ of a measurement with PDF $\func{p}{\bm{x};\theta}$ as
\[ \func{\mathcal{I}}{\func{p}{\bm{x};\theta}} := \left|\sqrt{\func{I}{\func{p}{\bm{x};\theta}}}\right|. \]
Likewise, we define the \emph{uncertainty} of an estimate as the square root of the estimate mean squared error
\[ \Delta \est{\theta} := \left|\sqrt{\func{MSE}{\est{\theta}}}\right|. \]
\end{definition}

\begin{remark}[Information in an estimate or the data?]\label{re:inf}
The Fisher information moves our attention from the estimate, to instead characterising the information contained in the probabilistic data-set (in a PDF). The estimate does not, in and of itself provide information on $\theta$, because the function $\func{est}{\cdot}$ does not depend on $\theta$. Rather, the estimate utilizes information provided by the data-set. This shift of focus to the information provided by a PDF allows us to obtain a better understanding of what the CRLB signifies.

The CRLB states that the information on $\theta$ provided by a data-set is uniquely defined, and the inverse of this quantity (nearly) uniquely limits the uncertainty of any estimate of $\theta$. We say nearly, because the Fisher information can be scaled by any non-negative number whilst still being a measure, and \cref{eq:CRLB} preserves this freedom. From this perspective the numerator of \cref{eq:CRLB} can be viewed as a scaling factor, one equal to the relative length of the estimate space and the parameter space\footnote{This relation only holds if the PDF is one-to-one and therefore requires \cref{ass:pdf}, otherwise the gradient can overestimate $\est{\left\Vert \Theta\right\Vert}$ and we have $\frac{\partial }{\partial \theta}\eval{\est{\theta}} \simleq \frac{\est{\left\Vert \Theta\right\Vert}}{\left\Vert \Theta \right\Vert}$.}
\be\label{eq:scale} \frac{\partial }{\partial \theta}\eval{\est{\theta}} \simeq \frac{\est{\left\Vert \Theta\right\Vert}}{\left\Vert \Theta \right\Vert}, \ee
where $\left\Vert \Theta \right\Vert$, $\est{\left\Vert \Theta\right\Vert}$ denote the length of the parameter and estimate space respectively, i.e. $\left\Vert \Theta \right\Vert = b - a$ for Euclidean distance. For example, if $\est{\theta} = \alpha$ for all $\theta$, then the estimate space is a single point with zero length and equivalently $\frac{\partial }{\partial \theta}\eval{\est{\theta}} = 0$. Whereas if $\est{\theta} = \theta$ for all $\theta$, then $\est{\theta}$ takes on every value of $\theta$, the spaces have identical lengths, and $\frac{\partial }{\partial \theta}\eval{\est{\theta}} = 1$.

\end{remark}

\Cref{eq:scale} motivates a specific scaling (choice of normalisation constant) for any estimate. We would like the estimate space to have the same length as the parameter space, to ensure that the estimate has the same ``units" as the parameter. Computationally this would also ensure that the algorithm returns a solution to every problem instance. In fact, the CRLB is usually presented in a form which places exactly this restriction on the estimate -- that it is unbiased.
\begin{definition}[Unbiased estimate]
An estimate of $\theta$ is defined as \emph{unbiased} if and only if it satisfies
\[ \eval{\est{\theta}} = \theta, \quad \forall \theta \in \Theta. \]
An unbiased estimate has no systematic error or bias, it is (on average) equal to the true value of $\theta$, the key qualifier being that this statement holds no matter the value of $\theta$, i.e. for all $\theta \in \Theta$. Given enough resources, we expect that an unbiased estimate should converge to $\theta$.
\end{definition}

\begin{genthm}{theorem}{Cram\'er-Rao lower bound for an unbiased estimate}
Suppose \crefrange{ass:classical}{ass:regularity} hold and that the estimate $\est{\theta}$ is unbiased. Then the estimate has a mean squared error (equiv. variance) of 
\be\label{eq:CRLBunbiased} \func{MSE}{\est{\theta}} = \func{Var}{\est{\theta}} \geq \frac{1}{-\eval{\frac{\partial^2 \func{log}{\func{p}{\bm{x};\theta}}}{\partial \theta^2}}} = \frac{1}{\eval{\left(\frac{\partial \func{log}{\func{p}{\bm{x};\theta}}}{\partial \theta}\right)^2}}, \quad \forall \theta \in \Theta. \ee
\end{genthm}
\begin{proof}
An unbiased estimator satisfies $\eval{\est{\theta}} = \theta$, so we have $\frac{\partial }{\partial \theta}\eval{\est{\theta}} = 1$. Substituting into \cref{eq:CRLB}, then \cref{eq:CRLBunbiased} follows immediately.
\end{proof}


\begin{remark}[Relevance and application to computational estimators]
The unbiased assumption may seem restrictive, but it is important to note that computations are in general unbiased. If a deterministic computation solves a problem on all inputs, then by definition it is unbiased. Furthermore, if a quantum computation finds a solution with high probability and assuming the parameter space is bounded, then the amount of bias is limited, meaning that efficient quantum computations are exponentially close to being unbiased.
\end{remark}

\begin{remark}[Bayesian estimation]
Bayesian estimation takes a different approach to resolving the scaling issue for a measure. Rather than restrict the estimator to be unbiased, Bayesian estimation focusses on evaluating the \emph{expected} mean squared error of the estimate. Now we perform two averages, first an expectation over the measurement PDF $\func{p}{\bm{x};\theta}$ and an additional expectation taken over the probability that $\theta$ takes on each value in $\Theta$. One issue with this approach is in justifying the PDF used to describe $\theta$ as a random variable. Often, a uniform probability distribution is assumed as this distribution has maximum entropy, but such an assignment does not set a lower bound, since taking any other distribution allows for better estimation.
\end{remark}


\subsubsection*{Summary}
The CRLB restricts the amount of information a measurement with a PDF described by $\func{p}{\bm{x};\theta}$ can provide on the value of an unknown parameter $\theta$. This information in turn bounds the uncertainty of any (unbiased) estimate of $\theta$. The CRLB says that if a measurement result has (on average) higher dependence on $\theta$, then it provides more information on $\theta$, whereas measurements that have weak dependence on $\theta$ provide little information on the value. In particular, the CRLB directs our focus when searching for strategies to obtain more information on $\theta$. The best strategy is to produce measurements with results that depend greatly on $\theta$. This interpretation is going to be a cornerstone of the quantum mechanical version of the CRLB.

The CRLB is both intuitive to understand and simple to state. Firstly, the Fisher information of a PDF characterises how much information a probabilistic measurement provides. Secondly, the information contained in a measurement bounds the uncertainty of any unbiased estimate. Although it is just a bound and may not be realizable, if we are presented with an analysis that claims to extract more information from a measurement with a distribution $\func{p}{\bm{x};\theta}$ (as evidenced by an unbiased estimator with lower uncertainty), then we can confidently say either; one of the assumptions has been violated or the analysis is incorrect\footnote{There is one area of quantum metrology which claims to do this, in clear violation of the CRLB -- quantum squeezing.}.\\

Until now we have just considered statistical analyses of classical data-sets -- those that are described by a collection of real numbers. We now move on to address the quantum mechanical version of the CRLB.

\subsection{Quantum mechanical statistical parameter estimation}
Note that the classical CRLB does not impose any physical restrictions on the output of a measuring device, in particular how much the device output can respond to a signal. If we want to characterise the precision of an arbitrary measuring device, then the CRLB only addresses half of the problem. Sure, once the measurement data is produced, the CRLB tells us how much information has been provided, but it says nothing about the form of the dataset produced by a device, i.e. how much the probability distribution can depend on $\theta$ in the first place\footnote{Actually, even before any data has been produced the CRLB tells us how much information \emph{could} be provided, so long as the response of the measuring device has been characterized.}. What is to stop a sensor (or computer) from achieving arbitrarily high precision by sampling from a distribution where $\frac{\partial \func{log}{\func{p}{\bm{x};\theta}}}{\partial \theta}$ goes to infinity?

Put another way, although the CRLB allows us to characterise the information provided by a measuring device, it does not specifically limit how much information the measuring device can provide. Given a measurement PDF parametrized by a signal, we have no way of knowing whether such a measuring device can be physically realised and whether we are applying the CRLB to a realistic system. It seems like there should be a physical law that addresses this issue.\\

Indeed there is one.\\

For any physical device, quantum mechanics places a constraint on $\left|\frac{\partial \func{log}{\func{p}{\bm{x};\theta}}}{\partial \theta}\right|$ and therefore the information provided by the device; leading to the quantum Cram\'er-Rao lower bound (qCRLB). The qCRLB allows us to consider a question much broader than that posed by the classical CRLB; and furthermore to answer it in full. Not only do we consider the information provided by a given data-set, but we also use quantum mechanics to determine how much a physical measuring device can respond to a signal; in effect restricting the form of the measurement PDF and thereby limiting the amount of information any device can provide. In particular, the question we consider is... 
\begin{question}
Given a quantum state $\ket{\psi_0}$ (possibly represented by a density matrix $\rho_0$), that interacts with a Hamiltonian $\Op{H}(\theta)$ parametrized by a fixed, classical, deterministic parameter $\theta$, and undergoes the transformation $\ket{\psi_0} \to \ket{\psi (\theta)}$. What is the minimum $\func{MSE}{\est{\theta}}$ for any estimate of $\theta$, using any measurement allowable by quantum mechanics (and assuming the state evolves under $\Op{H}(\theta)$ according to the Schr\"odinger equation). I.e. what bounds do the postulates of QM place on the $\func{MSE}{\est{\theta}}$ obtainable from a measurement of $\ket{\psi(\theta)}$?
\end{question}
The answer to this question is obtained by considering the family of all possible quantum states and quantum measurements. The result is a bound on $\left|\frac{\partial \func{log}{\func{p}{\bm{x};\theta}}}{\partial \theta}\right|$ for any measurement of any quantum state. We can then apply the CRLB to this classical data-set (arising from the best possible measurement on the optimal quantum state), and bound the $\func{MSE}{\est{\theta}}$ from such a data-set.\\

The recipe we use can be summarized as follows:
\begin{enumerate}
\item Start with an initial quantum state $\ket{\psi_0}$. We assume that $\ket{\psi_0}$ does not depend on $\theta$, i.e. we did not already sneak some information on $\theta$ into this initial state at time $t_0$. Taken together with \cref{ass:prior} that we have no other information of the value of $\theta$, this means that the only information on $\theta$ that we can physically obtain is through a measurement of the state $\ket{\psi_0} \to \ket{\psi(\theta,t)}$, and only after the initial state has evolved in response to some Hamiltonian.
\item The estimation problem, parametrized by $\theta$, is defined by a Hamiltonian $\Op{H}(\theta, t_0,t)$ (equiv. Unitary $\Op{U}(\theta,t_0,t)$). By \cref{ass:pdf}, the form of the Hamiltonian is perfectly known thus allowing the PDF to be characterised, albeit with an unknown value of $\theta$.
\item Using the postulates of QM: a) The Schr\"odinger equation -- which defines the state evolution under $\Op{H}(\theta, t_0,t)$, and b) The Born rule -- which defines the probabilities to measure any real valued data-set $\bm{x}$, corresponding to a collection of Hermitian measurement operators $\{\Op{X}\}$, we can place a bound on $\left| \frac{\partial}{\partial \theta}\func{log}{\func{p}{\bm{x};\theta,t}}\right|$ at time $t$.
\item Using a statistical estimation theorem on classical data-sets (the CRLB), we can bound the mean squared error of any unbiased estimate $\est{\theta}$, obtainable from any measurement described by a probability distribution $\func{p}{\bm{x};\theta,t}$.
\end{enumerate}

\begin{remark}
It is worth looking ahead here to note that in Point (3), $\left|\frac{\partial}{\partial \theta}\func{log}{\func{p}{\bm{x};\theta,t}}\right|$ is determined by how much the state $\ket{\psi(\theta, t)}$ responds to a change of $\theta$ in the Hamiltonian $\Op{H}(\theta, t_0,t)$. And how quickly a state can change in time is directly related the energy eigenvalues of the state. In fact, for a given Hamiltonian, $\left| \frac{\partial}{\partial \theta}\func{log}{\func{p}{\bm{x};\theta,t}}\right|$ is maximised when $\ket{\psi(\theta,t)}$ remains in an equal superposition of eigenstates with the greatest difference in eigenvalues (of $\partial \Op{H}(\theta, t_0,t)/\partial \theta$) for the entire evolution time. The Heisenberg limit in quantum metrology \cref{eq:HL}, is a direct consequence of this fact; the factor of $n$ is due to the $n$-fold greater energy difference as compared to a single qubit. More to the point, we do not even need the CRLB to rule out efficient quantum computation. A simple energy argument can be used to rule out the possibility of a state changing exponentially quickly in time, whereas efficient quantum computation requires the quantum state to follow a path that increases exponentially in time. Making this argument rigorous, however is more technical. 
\end{remark}

Before formulating a quantum mechanical version of the estimation problem, some definitions need to be introduced. As we will see, the following two definitions are critical.

\begin{definition}[A single quantum state vector]
A \emph{single (quantum) state vector} is defined as any non-separable unit vector in a complex Hilbert space $\mathscr{H}$. I.e. any normalised vector that cannot be decomposed into the tensor product of more than one vector in Hilbert spaces of lower dimension.

\emph{Clarification: }Empty state vectors are neglected in this definition. Just as the trivial decomposition of a prime number by a factor of 1 does not make the number composite, the tensor product of a single state vector with a trivial, unit dimensional Hilbert space does not produce a separable state vector. E.g. the tensor product of a single particle state vector with the empty/vacuum state is a single state vector\footnote{In optical interferometry, quantum states are often represented in an occupation number basis, and such a state would be written: $\ket{1}\otimes \ket{0}$.}.
\end{definition}
\begin{remark}
As hinted, there is a clear analogy between a single quantum state vector and prime numbers -- they both cannot be factored into smaller units. This suggests that we treat non-separable vectors as the atomic or indivisible units of Hilbert space and any physical system. In fact, this definition is central to our argument limiting the power of quantum computation. As entanglement is the single defining feature of a (pure state) quantum computation, we have already defined a metric in which quantum computers perform poorly -- the counting measure for the number of non-separable state vectors. If we can bound computational performance purely in terms of the number of state vectors, then we are well on the way. Our plan of attack is to decompose any physical device or computer into its constituent parts -- non-separable state vectors -- and consider the information provided by each non-separable state vector.
\end{remark}

\begin{definition}[A quantum measurement]
A \emph{quantum measurement} on a single quantum state vector is defined as a probabilistic function from a complex Hilbert space to a real probability space $\func{meas}{\cdot} : \mathscr{H} \to (\mathbb{R}, p)$, satisfying the following conditions (here assuming the measurement results are discrete).

A given measurement is described by a collection of Hermitian operators $\{\Op{X}_x\}$ in $\mathscr{H}$ that satisfy the completeness relation
\be \sum_{x\in \mathcal{X}} \Op{X}_x^{\dagger}\Op{X}_x = I,
\label{eq:residentity}\ee
with measurement outcomes $x \in \mathcal{X}$\footnote{Often assumed that $x$ is given by the eigenvalues of $\Op{X}_x$, and we need that $\mathcal{X}$ forms a $\sigma$-algebra.}. The measurement outcomes for a quantum measurement on a state $\ket{\psi}$ are observed with a probability given by the Born rule
\be \funcwrt{p}{x}{\ket{\psi}} := \bra{\psi}\Op{X}^{\dagger}_x \Op{X}_x \ket{\psi}. 
\label{eq:Born}\ee
\end{definition}

\begin{definition}[\hypertarget{Quantum estimation problem}{Quantum estimation problem}]
Consider a collection of single quantum states $\{\ket{\psi_0}\}$ which are used to measure a signal $\theta$ parametrized by a Hamiltonian $\Op{H}(\theta, t_0, t)$. We define a \emph{quantum estimation problem} as the task of estimating the unknown value of $\theta$, using \emph{only} the information provided by a series of measurements on $\{\ket{\psi (\theta, t)}\} = \Op{U}(\theta, t_0, t)\{\ket{\psi_0}\}$. We further allow that the quantum state evolution can be influenced by a control Hamiltonian $\Op{H}_c(t_0, t)$, that does not depend on $\theta$.
\end{definition}

\begin{remark}[Assumptions for quantum estimation]\label{ass:quantum}
The same assumptions used in the classical estimation problem (\cref{ass:classical} -- \cref{ass:regularity}) apply to the quantum estimation problem. Namely that $\theta$ is a scalar parameter of fixed definite value, that we have no additional information on $\theta$ that is not provided by the measurement, (possibly that the estimate $\est{\theta}$ is unbiased) and the measurement PDF satisfies the regularity conditions.
\end{remark}
Due to how the quantum estimation problem is formulated, we need to make \cref{ass:prior} more stringent, we further have to prevent one using an infinite amount of energy.
\begin{assumption}[Initial state]\label{ass:init_state}
We assume that at time $t_0$, the collection of single quantum states $\{\ket{\psi_0}\}$ does not depend on $\theta$, and therefore there is no information on $\theta$ already baked into the initial states. Mathematically, this assumption is expressed by the condition:
\[ \left\Vert \frac{\partial\ket{\psi_0}}{\partial \theta}\right\Vert = 0, \quad \forall \ket{\psi_0} \in \{\ket{\psi_0}\}.\]
\end{assumption}

\begin{assumption}[Bounded total energy]\label{ass:total_energy}
We assume for all time $t$, the total energy available for quantum evolution is finite and bounded by a constant
\[ \norm{\Op{H}(\theta,t_0,t)+\Op{H}_c(t_0,t)}\leq E_0,  \text{ for } E_0 \in \mathbb{R}_{\geq 0}.\]
\end{assumption}

Unfortunately, as the \hyperlink{Quantum estimation problem}{Quantum estimation problem} is currently formulated, it is still difficult to derive a general lower bound. Therefore we are going to add some caveats to make the analysis tractable. The two caveats needed for a rigorous formulation of the qCRLB are: 1) we restrict analysis to a single measurement, and 2) we restrict the measurement to that of a single quantum state. It turns out that these caveats are critical to obtaining a rigorous bound as they stop us from obtaining information on $\theta$ during $t$ and using that information to improve subsequent measurements. In short, \cref{ass:prior} implies that we cannot influence the evolution of $\ket{\psi_0}$ by taking advantage of any additional information on $\theta$, if however we can perform measurements at intermediate times, then this assumption no longer holds. To apply \cref{ass:prior} we need to ensure that the only information on $\theta$ we obtain is at the end of the experiment, after all evolution and from one single measurement. These restrictions allow us to show that quantum computers cannot efficiently solve an entire class of (estimation) problems.

\begin{assumption}[Single measurement on a single state vector]\label{ass:single}
We consider only the information provided by a single measurement on a single quantum state vector. In this case, the measurement result is a single number $x$ (not a vector) with PDF given by $\funcwrt{p}{x}{\ket{\psi}}$. This assumption is key to being able to restrict the performance of quantum computers as it prevents additional information on $\theta$ (obtained at an intermediate time) being used to improve the measurement result.
\end{assumption}

This additional restriction on the quantum estimation problem, which is critical to our analysis, ends up making it harder. 
\begin{definition}[Hard quantum estimation problem]
We define a \emph{hard} quantum estimation problem as a quantum estimation problem where we restrict to a single measurement on a single quantum state vector. I.e. one in which we enforce \cref{ass:single}.
\end{definition}

The restriction to a single state vector is important here, since with a collection of quantum states, we could measure some of the state vectors and obtain information on $\theta$, and then use this information to improve our measurements and control of the other state vectors.

\begin{remark}[A new computational model]
We can consider these assumptions as defining a new computational model. In this computational model, the computation ends with the computer in a single state vector and only a single measurement is performed on this state vector. At this point the computation finishes.
\end{remark}

We can now define the quantum Fisher information of a single quantum state vector and relate it to the Fisher information of a probability distribution.
\begin{definition}[Quantum Fisher information]
To any single quantum state vector $\ket{\psi}$, we can assign a non-negative real number called the \emph{quantum Fisher information} (QFI) on $\theta$ of the state vector. The QFI is defined as \cite{Braunstein1996}
\be \func{QFI}{\ket{\psi};\theta} := 4\left[\left(\frac{\partial\bra{\psi}}{\partial \theta}\right)\left(\frac{\partial\ket{\psi}}{\partial \theta}\right) - \left|\bra{\psi}\left(\frac{\partial\ket{\psi}}{\partial \theta}\right)\right|^2\right].
\label{eq:QFI} \ee
\end{definition}
The above is a slightly simpler and more explicit form of the QFI originally derived by Holevo which he wrote in terms of a symmetric logarithmic operator on mixed states \cite{Holevo1982} (see also Helstrom \cite{Helstrom1967}), it is related to the distance metric on quantum state vectors defined by Wootters \cite{Wootters1981}. Using the following relation
\[ \sum_x \frac{1}{\func{p}{x;\theta}}\left(\frac{\partial \func{p}{x;\theta}}{\partial \theta}\right)^2 = \sum_x \frac{1}{\func{p}{x;\theta}}\left(\frac{\func{p}{x;\theta} \partial \func{log}{\func{p}{x;\theta}}}{\partial \theta}\right)^2 = \eval{\left(\frac{\partial \func{log}{\func{p}{x;\theta}}}{\partial \theta}\right)^2}. \]
one can show that under the Born rule, no single quantum measurement of $\ket{\psi}$ can have a PDF with Fisher information greater than the QFI. I.e. denoting $\mathcal{M}$ as the class of allowable quantum measurements -- collections of Hermitian operators satisfying \cref{eq:residentity} and \cref{eq:Born} -- we have the following inequality \cite{Helstrom1967, Holevo1982, Braunstein1994, Braunstein1996}
\be \eval{\left(\frac{\partial \func{log}{\funcwrt{p}{x;\theta}{\ket{\psi}}}}{\partial \theta}\right)^2} \leq \func{QFI}{\ket{\psi};\theta}, \quad \forall \{\Op{X}_x\}\in \mathcal{M}.\ee
Equivalently
\be \funclims{sup}{\{\Op{X}_x\}\in \mathcal{M}}{\func{I}{\funcwrt{p}{x;\theta}{\ket{\psi}}}} \leq \func{QFI}{\ket{\psi};\theta}. \ee

We can summarize the above results in a simple expression. Since the second term in \cref{eq:QFI} is non-negative, the Fisher information (obtainable from any single measurement) of a single quantum state is bounded by how much the state responds to the signal
\be \func{I}{\funcwrt{p}{x;\theta}{\ket{\psi}}} \leq 4\left\Vert \frac{\partial \ket{\psi}}{\partial \theta}\right\Vert. \ee
The state response bounds the dependence of the measurement PDF on $\theta$ (for any measurement), which in turn sets a precision limit on any (unbiased) estimator of $\theta$. States that do not respond to a change in $\theta$ provide less information than states with a higher response. The quantum CRLB (for unbiased estimators) follows immediately.
\begin{genthm}{theorem}{Quantum Cram\'er-Rao lower bound}
Suppose the above assumptions hold, then the mean squared error (equiv. variance) of any unbiased estimator obtained from a single measurement of a single quantum state vector is bounded by
\be \func{Var}{\est{\theta}} \geq \frac{1}{\func{QFI}{\ket{\psi};\theta}} \geq \frac{1}{4\left(\frac{\partial\bra{\psi}}{\partial \theta}\right)\left(\frac{\partial\ket{\psi}}{\partial \theta}\right)}, \quad \forall \theta \in \Theta. \ee
\end{genthm}

\subsubsection*{Summary}
In short, we have shown how to bound the amount information one can extract from a single quantum state vector using just a \emph{single} measurement. The information on $\theta$, contained in a single quantum state vector is bounded by how much the state vector responds to the signal, i.e. $\left\Vert\partial\ket{\psi(\theta,t)}/\partial \theta\right\Vert$. By analysing parameters in explicit Hamiltonians, one can show that this leads to a general uncertainty bound as given in \cref{eq:HL}, see e.g. \cite{Giovannetti2004, Giovannetti2006, Zwierz2012, Pang2017, Degen2017, Pezze2018, Gorecki2020}.


\subsection{Issues with the CRLB}
It is worth remarking on two issues with the CRLB that I have never seen discussed, and if we are aiming at mathematical rigour are critical to address.

Firstly, the CRLB only addresses information provided by the data-set, however, we are also given information in the problem formulation when we are told the domain of $\theta$. Specifically, the set $\Theta$ gives us information on the value that $\theta$ can take. With no data, we can find an estimator (but not an unbiased estimator) with $\func{MSE}{\est{\theta}} \leq \left(\left\Vert \Theta \right\Vert/2\right)^2$ by letting $\est{\theta}$ be the center of the interval. The CRLB only addresses the information provided in the measurement PDF, but we are given further information, we are told the parameter space $\Theta$. Both of these sets are provided to us, and both provide information on the value of $\theta$.

To ensure the CRLB holds, \cref{ass:prior} prevents us from using this information on $\Theta$ in deriving an estimate, however it is very a restrictive assumption. If we can combine information measures on both of these sets, then we can derive a more general bound on the MSE of any estimate, and one which fully takes into account all the information we have access to. In fact, in coming up with the estimate $\est{\theta} = \alpha$ in \cref{re:inf}, we information on the parameter space to generate the estimate, and in fact violated \cref{ass:prior}. This is what makes the assumption so difficult to enforce, we must forget or throw away information provided to us in order to satisfy \cref{ass:prior}.

The second issue is that the CRLB is only valid with respect to Euclidean distance, whereas many metrics that we care about are not Euclidean. In particular, the distance between state vectors in Hilbert space is not Euclidean but is described by a Riemannian metric. The Fisher information is however a valid measure on Riemannian metrics. In computer science, many non-Euclidean distance metrics are commonly used, therefore we would like to derive an expression that bounds the mean squared error of an estimate in general metric spaces.



\subsubsection*{Conclusion}
We have derived a contradiction between the qCRLB and the outcome predicted for a single measurement on a quantum computer in a single entangled state vector. The contradiction can be summarized as follows. The qCRLB says that the information on a parameter $\theta$ contained in a single quantum state and observed in the measurement PDF cannot increase exponentially in time. Whereas, to perform efficient computation, the qCRLB says that the information (on some parameter) in a measurement PDF \emph{must} increase exponentially in time. 
Using measure theory, we can resolve this logical contradiction. There is no valid way to assign an information measure to a state in quantum mechanics which increases exponentially in time (unless we make the metric distance exponential or increase the energy exponentially). Thus, in a single run, a quantum computer in an entangled state vector cannot efficiently solve any computation problem that can be recast in terms of parameter estimation. We expect that this constitutes nearly the entire class of verifiable computational problems.

\bibliographystyle{naturemag}
\bibliography{references2025.bib}

\section{Appendix}
Explicit form of Fisher information for continuous and discrete random variables.
\[ \begin{cases}
\sum_{\bm{x} \in \mathcal{X}} \left(\frac{\partial \func{log}{\func{p}{\bm{x};\theta}}}{\partial \theta}\right)^2 \func{p}{\bm{x};\theta} = \sum_{\bm{x} \in \mathcal{X}} \left(\frac{\partial \func{p}{\bm{x};\theta}}{\partial \theta}\right)^2 \frac{1}{\func{p}{\bm{x};\theta}} \quad \text{if }\bm{x} \text{ is discrete},\\
\int_{\mathcal{X}} \left(\frac{\partial \func{log}{\func{p}{\bm{x};\theta}}}{\partial \theta}\right)^2 \func{p}{\bm{x};\theta} \mathrm{d}\bm{x} = \int_{\mathcal{X}} \left(\frac{\partial \func{p}{\bm{x};\theta}}{\partial \theta}\right)^2 \frac{1}{\func{p}{\bm{x};\theta}} \mathrm{d}\bm{x} \quad \text{if }\bm{x} \text{ is continuous}.
\end{cases} \]

\begin{definition}[\hypertarget{estimate variance}{Estimate variance}]
The (Euclidean) \emph{variance} of an estimate is defined as
\be \func{Var}{\est{\theta}} := \eval{\left(\est{\theta} - \eval{\est{\theta}}\right)^2},\ee
which is the average squared distance of the estimate from its expected value using Euclidean distance.
\end{definition}
\begin{definition}[Expected value of an estimate]
Expressed in terms of the estimator function $\func{est}{\bm{x}}$, the \emph{expected value} $\eval{\est{\theta}}$ of an estimate is defined as
\[ \eval{\est{\theta}} := \begin{cases}
\sum_{\bm{x} \in \mathcal{X}} \func{est}{\bm{x}}\func{p}{\bm{x};\theta} \quad \quad \text{if }\bm{x} \text{ is discrete},\\
\int_{\mathcal{X}} \func{est}{\bm{x}}\func{p}{\bm{x};\theta} \mathrm{d}\bm{x} \quad \quad \text{if }\bm{x} \text{  is continuous}.
\end{cases} \]
\end{definition}

\subsubsection{Proof of the Cram\'er-Rao lower bound}\label{proof:CRLB}
\begin{proof}[Proof of \cref{th:CRLB}] For the case that $\bm{x}$ is a continuous random variable (following the textbook of Kay \cite{Kay1993}). We use the following two identities
\be \label{eq:id2} \frac{\partial \func{p}{\bm{x};\theta}}{\partial \theta} = \func{p}{\bm{x};\theta} \frac{\partial \func{log}{\func{p}{\bm{x};\theta}}}{\partial \theta}.\ee
and 
\be \label{eq:id1} \eval{\frac{\partial \func{log}{\func{p}{\bm{x};\theta}}}{\partial \theta}}\cdot\theta = 0,\ee

Deriving the expected value of the estimate with respect to $\theta$, we obtain
\begin{align*}
&\frac{\partial }{\partial \theta}\int_{\mathcal{X}} \func{p}{\bm{x};\theta} \est{\theta} \; \mathrm{d}\bm{x} \quad \myeq{Ass.\,\ref{ass:regularity}}{=}\quad \int_{\mathcal{X}} \frac{\partial \func{p}{\bm{x};\theta}}{\partial \theta} \est{\theta} \; \mathrm{d}\bm{x} +\underbrace{ \int_{\mathcal{X}} \func{p}{\bm{x};\theta} \frac{\partial \est{\theta}}{\partial \theta} \; \mathrm{d}\bm{x}}_{=\,0\;\; Ass.\,\ref{ass:prior}} \quad \myeq{\cref{eq:id2}}{=}\quad \int_{\mathcal{X}} \func{p}{\bm{x};\theta} \frac{\partial \func{log}{\func{p}{\bm{x};\theta}}}{\partial \theta}\est{\theta} \;\mathrm{d}\bm{x}\\
&\;\quad\myeq{\cref{eq:id1}}{\implies}\quad \int_{\mathcal{X}} \func{p}{\bm{x};\theta} \frac{\partial \func{log}{\func{p}{\bm{x};\theta}}}{\partial \theta}\left(\est{\theta}-\theta\right) \;\mathrm{d}\bm{x} = \frac{\partial }{\partial \theta}\eval{\est{\theta}}.
\end{align*}

\Cref{eq:CRLB} follows immediately from application of the Cauchy-Schwarz inequality\footnote{Kay notes that it holds with equality if and only if $\func{g}{\bm{x}} = c \func{h}{\bm{x}}$ for $c$ some constant not dependent on $\bm{x}$. The functions $\func{g}{\cdot}$ and $\func{h}{\cdot}$ are arbitrary functions, while $\func{w}{\bm{x}}\geq 0$ for all $\bm{x}$.}
\[ \left(\int_{\mathcal{X}} \func{w}{\bm{x}}\func{g}{\bm{x}}\func{h}{\bm{x}} \; \mathrm{d}\bm{x} \right)^2 \leq \int_{\mathcal{X}} \func{w}{\bm{x}}\left(\func{g}{\bm{x}}\right)^2\; \mathrm{d}\bm{x} \cdot \int_{\mathcal{X}} \func{w}{\bm{x}}\left(\func{h}{\bm{x}}\right)^2\; \mathrm{d}\bm{x},\]
with $w(\bm{x}) = \func{p}{\bm{x};\theta}$, $g(\bm{x}) = \est{\theta}-\theta$, $h(\bm{x}) = \frac{\partial\func{log}{\func{p}{\bm{x};\theta}}}{\partial \theta}$.
\end{proof}
For the case that $\bm{x}$ is a discrete random variable, the proof follows analogously (see \cite{Cramer1946}).



\subsubsection{Measure theory}\label{app:measure}
\begin{definition}[\hypertarget{Measure}{Measure}]
\emph{Assigning a set to a non-negative number, that obeys additivity}.\\
Let $U$ be a set and $\mathcal{U}$ be a family of subsets of $U$ such that $\mathcal{U}$ forms a $\sigma$-algebra. A \emph{measure} on $\mathcal{U}$, $\func{meas}{\cdot}$, is a mapping $\mathcal{U} \to \extreal$, that assigns to each subset of $U$, (i.e. $S \in \mathcal{U}$), one and only one non-negative number. This makes $\func{meas}{\cdot}$ a function (on sets). To be a measure, the function, $\func{meas}{\cdot} : \mathcal{U} \to \extreal$ must satisfy the following axioms.
\begin{enumerate}
\item (Non-negativity) For all $S \in \mathcal{U}$, $\func{meas}{S} \geq 0$.
\item(Empty set) $\func{meas}{\varnothing} = 0$.
\item (Countable additivity) For any countable collection of disjoint sets $S_1, S_2, \dotsc \in \mathcal{U}$, then 
\[\func{meas}{\bigcup^{\infty}_{n = 1} S_n} = \sum^{\infty}_{n = 1} \func{meas}{S_n}.\]
\end{enumerate}
\end{definition}

\begin{proof}[\hypertarget{Proof}{Proof} of \cref{prop:FI}]
\emph{Fisher information is a measure with respect to $\theta$ on probability distributions.}
That the Fisher information of any set is non-negative and the Fisher information of the empty set is zero is immediately clear from the definition. Using the following relation for the outcome of a series of $R$ measurements, each described by the PDF $\func{p}{x_i;\theta}$
\[ \func{p}{\bm{x};\theta} = \func{p}{\{x_1, x_2, \cdots, x_R\};\theta} = \prod^R_{i=1} \func{p}{x_i;\theta},\]
then we can show countable additivity is satisfied by the following relation for the Fisher information of $R$ measurements
\be \func{I}{\func{p}{\bm{x};\theta}} = \sum^R_{i=1} \func{I}{\func{p}{x_i;\theta}}.\ee
Meaning that each measurement PDF corresponds to a disjoint set in the total probability space.
\end{proof}


\end{document}